\documentclass[aps,nofootinbib,preprintnumbers,citeautoscript]{revtex4}
\usepackage{amsmath,amssymb}
\usepackage{enumerate}
\usepackage{graphicx}
\usepackage{amsthm}
 \usepackage{mathtools}
 \usepackage{textcomp}
  \usepackage{braket}
  \usepackage{url}
  \usepackage[flushleft]{threeparttable}
 \usepackage{makecell}
\usepackage{pifont}

\begin{document}
\title{RNA-2QCFA: Evolving Two-way Quantum Finite Automata with Classical States for RNA Secondary Structures} 

\author{Amandeep Singh Bhatia$^{1,^*}$,  Shenggen Zheng$^{2}$\\
$^{1}$\textit{Chitkara University Institute of Engineering \& Technology, Chitkara University, Punjab, India}\\
$^{2}$\textit{Center for Quantum Computing, Peng Cheng Laboratory, Shenzhen, China} \\
E-mail: $^{1,^*}$amandeepbhatia.singh@gmail.com}

\begin{abstract}
Recently, the use of mathematical methods and computer science applications   have got significant response among biochemists and biologists to modeling the biological systems. The computational and mathematical methods have  an enormous potential for modeling the deoxyribonucleic acid (DNA) and ribonucleic acid (RNA) structures. The modeling of DNA and RNA secondary structures using  automata theory had a significant impact in the fields of computer science. It is a natural goal to model the RNA secondary biomolecular structures using quantum computational models.  Two-way quantum finite automata with classical states are more dominant than two-way probabilistic finite automata in language recognition. 
The main objective of this paper is on using two-way quantum finite automata with classical states to simulate, model and analyze the ribonucleic acid (RNA) sequences.\\

\end{abstract}
\maketitle

\maketitle



\newcommand{\myqed}{\rule{2pt}{1em}}
\newenvironment{myproof}{\begin{proof}}{\let\qedsymbol\myqed\end{proof}}
\theoremstyle{plain}
\newtheorem{thm}{Theorem}
\theoremstyle{definition}
\newtheorem{defn}{Definition}
\newtheorem{exmp}{Example}[section]
\section{Introduction \& Motivation}
In recent years, the field of bioinformatics has gained much attention among research and academia communities to develop intelligent systems for simulating and analysis of molecular biology. Bioinformatics is a active, diverse and fast-growing research field.  It is the application of information technology to store, process and analyse the biological data, especially DNA, RNA, and protein sequences \cite{2}. Presently, the focus is on developing probabilistic models to examine the biological sequences at genome level. Till now, several methods based on automata theory, grammatical formalism, learning theory and statistical theory have been introduced to modeling and analyzing the behaviour of RNA, DNA and protein sequences \cite{3}. The accurate modelling and prediction of the genomewide are the major challenges in bioinformatics \cite{4}. 

Nowadays, quantum computing is the new buzz word and become a hot research topic in industry, academic and R\& D centres getting significant response and financial support from all directions. Quantum computing incorporates elements from  physics, mathematics and computer science \cite{5}. Quantum computers have the potential to tackle the problems that would take classical computers millions of years. It promises to solve complex real-world problems such as modeling financial risks, simulating chemistry and optimising supply chains. It would impact biologists to study the possible ways to interact and fold proteins with one another and chemists to model interactions between drugs. We can examine some properties quickly using quantum superposition principle and entanglement than with classical means \cite{6}. 

The finite automata theory is one of the keystone of theoretical computer
science \cite{7}. The combination of quantum mechanics and classical automata theory gives us quantum finite automata (QFA) \cite{8, 9}. The concept of QFA models was introduced by Moore and Crutchfield \cite{50} and Kondacs and Watrous \cite{12} separately, soon after the discovery of Shor’s factoring algorithm \cite{10}. QFA are abstract models of machines with finite memory for quantum computers, which play an important role in carrying out computation in real-time, i.e. the tape head takes exactly one step per input symbol and moves towards the right direction only \cite{11}. It is described as a quantum analogue of a classical finite automaton. It lays down the perception of quantum processors for executing quantum operations on reading the inputs.

Since then, there is a diversity of quantum automata models have been studied  and investigated in all directions such as quantum finite automata, Latvian QFA (LQFA) \cite{12}, 1.5-way QFA \cite{13}, two-way QFA (2QFA) \cite{12}, quantum pushdown automata (QPDA) \cite{55}, quantum Turing machine (QTM) \cite{15}, quantum multihead finite automata (QMFA) \cite{17, 40}, multi-letter QFA \cite{16}, one-way quantum finite automata with classical states (1QCFA) \cite{14}, two-way quantum finite automata with classical states (2QCFA) \cite{1}, quantum queue automata \cite{100} and many more since last two decades. These models are effective in examine the frontiers of computational properties and expressive power of automata. Quantum computers are more powerful than probabilistic Turing machines and even Turing machines. Therefore, quantum computational models can be consider as generalizations of its physical models \cite{8}. 

Bioinformatics introduces and utilizes biological computational algorithms for interpretation of biological processes based on  interaction between genomes \cite{18}. The biological sequences are modeled using grammatical formalism to efficiently solve the  bioinformatics computational problems such as prediction and classification of sequences, calculate multiple alignments, sequences analysis and data mining \cite{19}. DNA can be seen as recipe of an organism. It is double stranded and made from four different monomers  called nucleotides \textit{(A, C, G ,T)} representing  adenine (\textit{A}), cytosine (\textit{C}), guanine (\textit{G}) and thymine (\textit{T}). RNA is like DNA except the base  thymine (\textit{T}) is replaced by base uracil (\textit{U}). It is often single-stranded structure and folds around itself \cite{34}. Some of the bases form mismatched nucleotides, which results in formation of loops of unpaired single strands at the center or end of a duplex. 

Kondacs and Watrous \cite{12} proposed the notion of two-way quantum finite automata i.e. quantum version of two-way deterministic finite automata. It has been proved that it is more dominant than classical counterparts for language recognition. 2QCFA can recognize some context-free, context sensitive languages and all regular languages. But, atleast $O(log~n)$ qubits are needed to store positions of the input tape head, where \textit{n} is the length of an input string. In order to get over the disadvantage of 2QCFA, Ambainis and Watrous \cite{1} introduced two-way quantum finite automata with classical states (2QCFA). Its computational power lies in between 1QFA and 2QFA, but still more powerful than its classical variants. It has been proved that 2QCFA is more powerful than two-way probablistic finite automata (2PFA). It can recognize the palindrome language $L=\{ww^r \mid w \in \{a, b\}^*\}$, where $w^r$ is the reverse of \textit{w} in exponential time with one-sided error, but 2PFA cannot be designed for \textit{L} with bounded error.  Zheng et al. \cite{78} studied the state succinctness of 2QCFA. Zheng et al. \cite{20} investigated that 2QCFA can be designed for $L=\{xcy \mid x,y \in \{a, b\}^*, |x|=|y|, \Sigma=\{a, b, c\}\}$ with bounded error in polynomial time, but 2PFA takes an exponential time with bounded error. Qiu et al. \cite{21} proved that the class of languages recognized by 2QCFA are closed under union, intersection, complement and the reversal operation. 

Motivated from the above-mentioned facts, we have transcribed RNA secondary structures in the form of formal languages and modeled them using two-way quantum finite automata with classical states. The main objective is to examine how RNA secondary structures perform sequence identification equivalent to quantum automata models. The crucial advantage of this approach is that chemical reactions in the form of accept/reject signatures can be processed in linear time with one-sided bounded error.  The organization of rest of this paper is as follows: Subsection is devoted to prior work. In Sect. 2, some preliminaries are given. The notion of two-way quantum finite automata with classical states is given in Sect. 3. In Sect. 4, the RNA secondary structures (hairpin loop, pseudoknot and dumbbell structures) are transcribed in formal languages and modeled using two-way quantum finite automata with classical states. Finally, Sect. 5 is the conclusion.
\section{Prior Work}
During the last three decades, several representations of RNA and DNA sequences using automata theory and formal grammar have been found in literature. The structures of DNA and RNA are represented using the concept of classical automata theory. In 1984, Brendel and Busse \cite{22} transcribed nucleic acid sequences as words over the input alphabet of nucleotides and formulated that genomes can be described in formal language theory. Sung  \cite{88} modelled the RNA pseudoknots using context sensitive grammar. In 1992, Searls \cite{3} formulated RNA and DNA sequences such as pseudoknot, inverted and tandem repeat using indexed grammar. Later, Searls \cite{23} used string variable grammar to represent DNA sequences. 

Roy et al. \cite{24} proposed the concept of micron automata processor to find the conserved sequences in protein or multiple DNA sequences. Cai et al. \cite{25} represented the pseudoknot biomolecular structures of RNA using parallel communicating grammar. Barjis et al. used finite automata as modeling tool for formulation and simulation of production of proteins. Mizoguchi et al. \cite{26} modeled different classes of pseudoknot structure with  stochastic multiple context-free grammar. Kuppusamy and Mahendran \cite{27} presented the notion of matrix insertion-deletion system and used to analyze and model the RNA secondary structures such as stem and loop, pseudoknot, attenuator, internal loop, bulge loop and kissing hairpin. Recently, Bhatia and Zheng \cite{101} transcribed the chemical reactions in from of formal languages and modeled them using two-way quantum finite automata.  Fernau et al. \cite{28} introduced small size universal matrix insertion grammar to simulate the computation of DNA. Krasinski et al. \cite{29}  described the restricted enzyme in DNA with circular mode pushdown automata. 

Khrennikov and Yurova \cite{30} modeled the behavior of protein structures using classical automata theory and investigated the resemblance between the quantum systems and modeling behavior of proteins. Quantum omega automata can be used to model the behavior of chemical reactions in biological systems \cite{104}. Soreni et al. \cite{31} shown programmable three symbol three state finite automata and carried out biomolecular computations parallelly on surface. Lin and Shah \cite{32} represented the patterns in DNA sequences using statistical finite automata. Cavaliere et al. \cite{33} and  Rothemund used Turing machine and pushdown automata to analyze the action of a restricted enzyme in DNA. Bhatia and Kumar \cite{34} shown the modeling of double helix, hairpin and internal loops using linear time 2QFA.  Recently, Duenas-Diez and Perez-Mercader \cite{35} designed molecular machines for chemical reactions. It has been demonstrated that chemical reactions transcribed in formal languages, can be recognized by Turing machine without using biochemistry.
\section{Two-way Finite Automata with Quantum and Classical States}
Ambainis and Watrous \cite{1} presented the notion of two-way quantum finite automata with classical states (2QCFA). The computational power of 2QCFA lies in between 2QFA and 1QFA. In 2QCFA, the tape head position is classical and the internal state may be a (mixed) quantum state.
\begin{defn} \textup{\cite{1}}
	A 2QCFA is defined as a nonuple $(S, Q, \Sigma, \Theta,  \delta,q_0, s_0, S_{acc}, S_{rej})$, where
	\begin{itemize}
		\item $\it S$ is a finite set of classical states,
		\item $\it Q$ is a finite set of quantum states,
		\item $\Sigma$ is an input alphabet such that $\Gamma=\Sigma \cup \{\#, \$\}$, where \# and \$ are left and right-end markers, respectively,
		\item $\Theta$ defines the evolution of the quantum portion of the internal state, 
		\begin{equation}
		S \setminus (S_{acc} \cup S_{rej}) \times \Gamma \rightarrow U(\mathcal{H}(Q)) \cup P(\mathcal{H}(Q)) 
		\end{equation}
		where $P(\mathcal{H}(Q))$ and $U(\mathcal{H}(Q))$ representing the projective measurements and unitary operators over Hilbert space $\mathcal{H}(Q)$ with set \textit{Q}. Therefore, $\Theta(s, \vartheta)$ is equivalent to either projective measurement or a unitary evolution. 
		\item $\delta$ defines the evolution of classical states. If $\Theta(s, \vartheta) \in P(\mathcal{H}(Q))$, then $\delta$ is defined as
		\begin{equation}
		S \setminus (S_{acc} \cup S_{rej}) \times \Gamma \times E \rightarrow S \times \{\leftarrow, \uparrow, \rightarrow \},
		\end{equation}		
		where \textit{E} denotes possible set of eigenvalues $E=\{e_1, e_2,..., e_n\} $ and projector set $\{P(e_j): j$=1,2,...,\textit{n}\} and $P(e_j)$ is the projector onto the eigenspace. $\{\leftarrow, \uparrow, \rightarrow \}$ shows the head movement towards left, stationary and right side of the input tape, respectively. If $\Theta(s, \vartheta) \in U(\mathcal{H}(Q))$, then $\delta$ is defined as
		\begin{equation}
		S \setminus (S_{acc} \cup S_{rej}) \times \Gamma  \rightarrow S \times \{\leftarrow, \uparrow, \rightarrow \},
		\end{equation}		
		\item $q_0$ is an initial quantum state $q_0 \in Q$,
		\item $s_0$ is an initial classical state $s_0 \in S$,
		\item $S_{acc}, S_{rej}$ are the set of accepting and rejecting states respectively, $(S_{acc},S_{rej}\subseteq S)$.
	\end{itemize}
\end{defn}
The computation procedure of 2QCFA to process a given input string \textit{w}   is as follows:  Initially, the classical and quantum state are $s_0$ and $\ket{q_0}$ and head is positioned on left-end marker \$. The state of 2QCFA is changed according to $\Theta(s_0, \$)$. On reading the input symbol $\sigma \in \Sigma$, the classical state is changed according to $\delta(s, \sigma)$ and quantum state is changed according to $\Theta(s, \sigma)$. The tape head is moved according to direction $d=\{\leftarrow, \uparrow, \rightarrow\}$. 
\begin{itemize}
	\item
	If $\Theta(s_0, \sigma)= U(\mathcal{H}(Q))$, then the classical state $s_0$ is transformed to $s_1$ according to $\delta(s_0, \sigma)=(s_1, d)$, quantum state is transformed as $\ket{q_0}=U\ket{q_0}$ and head movement is determined by \textit{d}. 
	
	\item If $\Theta(s_0, \sigma)= P(\mathcal{H}(Q))$, then the projective measurement is carried out on $\ket{q_0}$. Suppose, $P=\{P_1, P_2,..., P_n\}$  with possible eigenvalues $\{e_j\}^{n}_{j=1}$. After performing the measurement, we get a result $e_j \in E$ with probability $p_j=\bra{q_0}P_j
	\ket{q_0}$ and the quantum state is transformed as $P_j\ket{q_0}/ \sqrt{p_j}$. The classical state is changed as  $\delta(s_0, \sigma)= (s_1, d)$. 
\end{itemize}

Finally, each measurement outcomes are probabilistic and the classical state transitions may be also probabilistic. Thus, the 2QCFA is said to be accepted with probability  $S_{acc}(w)$  when the computation is halted and automata enters the classical accepting state $S_{acc}$, otherwise it is said to be rejected with probability $S_{rej}(w)$. Consider a language $L \subset \Sigma^*$,  it is said to be accepted with one-sided error $\epsilon$  by 2QCFA $M_Q$ if the probability of acceptance $Pr[M_Q~ accepts~w]=1$, $\forall w \in L$, otherwise it is said to be rejected with one-side error if $Pr[M_Q~rejects~w]\geq 1-\epsilon$ if $w \notin L$. 
\section{RNA Secondary Structures Modeling}
In this section, we analyze, model and simulate the RNA secondary biomolecular structures such as hairpin loop, pseudoknot structure and dumbbell structure using two-way quantum finite automata with classical states (2QCFA). We assume that the reader is familiar with the classical automata theory and the concept of quantum computing; otherwise, reader can refer to the theory of automata \cite{7}, quantum information and computation \cite{62, 64}.

\subsection{Hairpin Loop}
Hairpin is the primary unit secondary structure in RNA molecules. It plays a crucial  role in various  biological processes such as DNA transposition, DNA recombination, gene expressions and RNA-protein recognition. A hairpin loop consists of a base-paired stem formed in single-stranded nucleic acids and ends to form unpaired nucleotide bases \cite{36}. It is named according to size and composition of loop.  Fig 1 shows the representation of hairpin loop structure.  Hairpin loop can be transcribed as language $L_{h}=\{x \in \{a, u, g, c\}^* \mid x=x^r\}$ to form base pairing, where $x^r$ is \textit{x} in the reverse order.
The detailed proof for palindrome language $L=\{w \in \{a, b\}^* \mid w=w^r\}$ can be find in  Ambainis and Watrous \cite{1} paper.
\begin{figure}[h]
	\centering
	\includegraphics[scale=0.65]{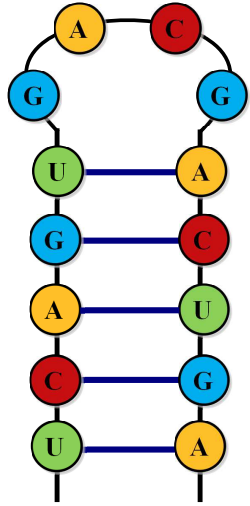}
	\caption{Representation of Hairpin loop structure}
\end{figure}
\begin{thm}
	A language $L_{h}=\{x \in \{a, u, g, c\}^* \mid x=x^r\}$, where $x^r$ is \textit{x} in  the reverse order, representing hairpin loop biomolecule structure can be recognized by 2QCFA with one-sided error in exponential time, which cannot be recognized by 2PFA. 
\end{thm}
\begin{proof}
	\begin{equation}
	U_a=U_u=\dfrac{1}{5} \begin{pmatrix}
	4 & 3 & 0 \\
	-3 & 4 & 0\\
	0 & 0 &  5
	\end{pmatrix}, U_g=U_c=\dfrac{1}{5} \begin{pmatrix}
	4 & 0 & 3 \\
	0 & 5 & 0\\
	-3 & 0 & 4
	\end{pmatrix}\end{equation}
	The idea of the proof is as follows. It consists of two phases. First we define a 2QCFA for $L_1$ using quantum register consisting three orthogonal states. In second phase, we modify the 2QCFA such as natural mapping is performed from three-dimensional Euclidean space to unit sphere in two-dimensional Hilbert space $\mathcal{H}$. We construct a 2QCFA $M_1$ for the language $L_{h}=\{x \in \{a, u, g, c\}^* \mid x=x^r\}$ with three quantum states $\{q_0, q_1, q_2\}$, where $q_0$ is an initial state. $M_1$ has two unitary matrices  $U_a=U_u$ and $U_g=U_c$ defined as follows. The automaton $M_1$ proceeds as follows:
	\begin{table}[!ht]
		\centering
		\caption{Details of the 2QCFA for $L_1$}
		\begin{tabular}{ p{0.5cm} p{13cm}}
			\hline
			\multicolumn{2}{l}{Repeat the following endlessly:} \\
			1. & Set the initial quantum state $q_0$ and shift the input tape head under the first input symbol.\\
			2. & While the presently symbol read is not \$, do the following: \\
			~ ~ (2.1). &  ~ ~ If the presently examined symbol is \textit{a} or \textit{u}, apply $U_a$ or $U_u$ on the quantum state, respectively.\\
			~ ~ (2.2). &  ~ ~ If the presently read symbol is \textit{g} or \textit{c}, execute $U_g$ or $U_c$ on the quantum state, respectively.\\
			~ ~ (2.3). &  ~ ~ Shift the position of tape head one square towards the right. \\
			3. & Repeat the following subroutines:\\
			~ ~ (3.1). &  ~ ~ Move the input tape head towards the left direction until the right-end marker symbol \# is reached.\\
			~ ~ (3.2). &  ~ ~ Shift the position of tape head one square towards the right.\\
			4. & While the presently symbol read is not left-end marker \$, do the following:\\
			~ ~ (4.1). &  ~ ~ If the presently read symbol is \textit{a} or \textit{u}, apply $U_{a}^{-1}$ or $U_{u}^{-1}$ on the quantum state, respectively.\\
			~ ~ (4.2). &  ~ ~ If the presently scanned symbol is \textit{g} or \textit{c}, execute $U_{g}^{-1}$ or $U_{c}^{-1}$ on the quantum state, respectively.\\
			~ ~ (4.3). &  ~ ~ Shift the position of the input tape head towards the right. \\
			\multicolumn{2}{l}{Perform the measurement on quantum state, if the outcome is not $q_0$, then it is rejected.}\\
			\multicolumn{2}{l}{Initialize the variable \textit{z}=0}\\
			5. & While the presently symbol scanned is not right-end marker \#, do the following:\\
			~ ~ (5.1). &  ~ ~ Replicate \textit{k} coin flips. Initialize \textit{z}=1, if all outcomes are not "heads".\\
			~ ~ (5.2). &  ~ ~  Shift the position of tape head one square towards the left.\\
			\multicolumn{2}{l}{if \textit{z}=0, it is said to be accepted.}\\
			\hline
		\end{tabular}
	\end{table}
	Consider an input string $w=w_1,w_2,...,w_n$, the tape squares are indexed by 0 and \textit{n}+1 consist both end-markers \# and \$, respectively. The computation process of $M_1$ starts with quantum state $\ket{q_0}$ is as follows. As while-loop 2 is implemented, the input tape head traverse each input symbol and performs $U_a$ or $U_u$ on the quantum state (depends upon whether the input symbol is \textit{a} or \textit{u} and performs $U_g$ or $U_c$ on the quantum state (depends upon whether the input symbol is \textit{g} or \textit{c}, respectively. Suppose $W_i$ denote the matrix ($U_a$ or $U_u$) and  ($U_g$ or $U_c$), as defined in (4), depending upon $w_i$ is (\textit{a} or \textit{u}) and (\textit{g} or \textit{c}). The automaton $M_1$ changes its state after executing loop 2 as
	\begin{equation}
	\beta_0 \ket{q_0}+\beta_1 \ket{q_1}+\beta_2 \ket{q_2}
	\end{equation}
	for $(\beta_0, \beta_1, \beta_2)^T$=${\dfrac{1}{5}}^n W_n...W_1 (1, 0, 0)^T$. As we repeat the subroutines 3, the input tape head is moved towards the left direction until the right-end marker is read, then shift the tape head one position to the right. Now, in loop 4, the inverses of $U_a$ and $U_g$ are performed, the quantum state is changed as
	\begin{equation}
	\gamma_0 \ket{q_0}+\gamma_1 \ket{q_1}+\gamma_2 \ket{q_2}
	\end{equation}
	for $(\gamma_0, \gamma_1, \gamma_2)^T$=$W_n^{-1}...W_1^{-1}W_n...W_1(1, 0, 0)^T$. The quantum state is measured and $M_1$ is said to be rejected with probability $P_{rej}=\gamma_1^{2}+\gamma_2^{2}$, else it collapses to initial state $q_0$. If \textit{w} is a palindrome, then $P_{rej}=0$, else $P_{rej} > 25^{-n}$. Initialize the variable \textit{z} equal to 0, which is stored in classical state. Finally, on executing the while-loop 5, the input string is said to be accepted if the loop is terminated with $z=0$. It is known from standard result in probability theory that probability of reaching the location \textit{n}+1 is $\dfrac{1}{n+1}$. On flipping the \textit{k} coins, the probability of acceptance $P_{acc}=1/2^k(n+1)$.  If the algorithm is repeated indefinitely, then the probability of rejectance is 
	\begin{equation}
	Pr[M_1 ~ rejects ~ w]= \sum_{i \geq 0} (1-P_{acc})^i (1-P_{rej})^iP_{rej}=\dfrac{P_{rej}}{P_{acc}+P_{rej}-P{acc}P_{rej}}
	\end{equation}
	and accepting probability is
	\begin{equation}
	Pr[M_1 ~ accepts ~ w]= \sum_{i \geq 0} (1-P_{acc})^i (1-P_{rej})^{i+1}P_{acc}=\dfrac{P_{acc}-P_{acc}P_{rej}}{P_{acc}+P_{rej}-P{acc}P_{rej}}
	\end{equation}
	If the input string $w \in L_h$, then the probability of $M_1$ accepting \textit{w} is 1. Suppose $k\geq\text{max}\{log~25, -log~\epsilon\}$, it can be checked that if $w \notin L_h$, the $M_1$ rejects the input string with probability atleast $1-\epsilon$.
\end{proof}

\subsection{Pseudoknot Structure}
A pseudoknot structure is a double-hairpin structure that forms an extended quasi-continuous helix structure and double connecting loops \cite{37}. It is formed when pairs are created between the bases outside and inside of a hairpin or internal loop. Pseudoknot structure plays a crucial role in RNA functions such as regulation of splicing and translation and ribosome frameshifting \cite{38}. It is considered as a key component of ribozymes or ribosomal RNAs.  Fig 2 describes the pseudoknot secondary structure. A closer look at pseudoknot structure shows a similarly with constructs of natural language (i.e. dependencies are forced to cross) such as $\{a^ng^mu^nc^m \mid n, m \geq 1\}$  \cite{61}. Thus, the number of \textit{a}'s is equal to the number of \textit{u}'s and correspondingly the number of \textit{g}'s is equal to the number of \textit{c}'s.
\begin{figure}[h]
	\centering
	\includegraphics[scale=0.65]{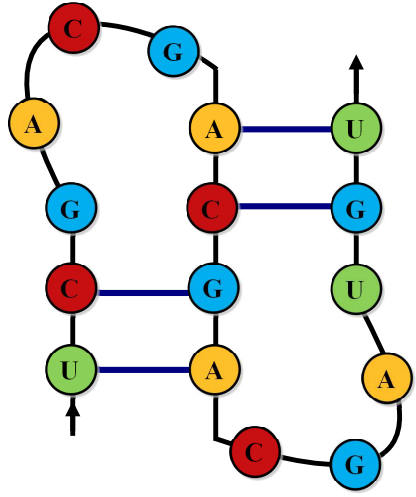}
	\caption{Representation of psedoknot bio-molecular structure}
\end{figure}
\begin{thm}
	A language $L_{p}=\{a^ng^mu^nc^m \mid n, m \geq 1\}$ representing pseudoknot  biomolecule structure can be recognized by 2QCFA  in polynomial time with error probability $\epsilon=\epsilon_1+\epsilon_2-\epsilon_1\epsilon_2$. 
\end{thm}
\begin{proof}
	The idea of the proof is as follows. The construction of 2QCFA $M_2$ consists of three phases. Firstly, it checks whether the input is in form $a^{+}u^{+}g^{+}c^{+}$, if not, then the input string is said to be rejected. In second phase, we check that it is in form $L_1=\{a^ng^*u^nc^* \mid n \geq 1\}$ or not. If yes, then we check that it is in form $L_2=\{a^*g^mu^*c^m \mid m \geq 1\}$. Since, 2QCFA can be designed for $L_{eq}=\{a^nb^n \mid n \geq 1\}$. Similarly, it can recognize $L_1$ and $L_2$ in polynomial time. 
	Qiu \cite{21} proved that the class of languages recognized by 2QCFA are closed under intersection, reversal, complement and union operations. For convenience, 2QCFA$_{\epsilon}$(\textit{poly-time}) notation is used to denote the class of languages recognized by 2QCFA in polynomial time with error probability $\epsilon \geq 0$.  Therefore, if $L_1 \in$ 2QCFA$_{\epsilon_1}$(\textit{poly-time}) and  $L_2 \in$ 2QCFA$_{\epsilon_2}$(\textit{poly-time}), then $L_p=L_1 \cap L_2$ can be recognized by 2QCFA in polynomial time with $\epsilon=\epsilon_1+\epsilon_2-\epsilon_1\epsilon_2$. 
	
	Based on the above analysis, the prove of this theorem is described now more formally. Let 
	\begin{equation}
	M_j=(S_j, Q_j, \Sigma_j, \Theta_j,  \delta_j, q_{j, 0}, s_{j, 0}, S_{j, acc}, S_{j, rej})
	\end{equation}
	be 2QCFA's for recognizing $L_j$ in polynomial time with error probabilities $\epsilon_j \geq 0$ for \textit{j}=1, 2, where
	\begin{itemize}
		\item $S_j=\{s_{j, 0}, s_{j, 1}, ..., s_{j, m_j}\}$, 
		\item $Q_j=\{q_{j, 0}, q_{j, 1}, ..., q_{j, n_j}\}$
	\end{itemize}
	We construct a 2QCFA for $L_p$ such that $M_p=(S, Q, \Sigma, \Theta,  \delta, q_0, s_{0}, S_{acc}, S_{rej})$ where
	\begin{itemize}
		\item $S=S_1 \cup S_2 \cup \{z_{1, i} \mid i= 0, 1, ..., n_1\}$, 
		\item $Q =Q_1 \cup Q_2$,
		\item $\Sigma=\Sigma_1 \cup \Sigma_2$,
		\item $q=q_{1, 0}$,
		\item $s=s_{1, 0}$,
		\item $S_{acc}=S_{1, acc} \cup S_{2, acc}$,
		\item $S_{rej}=S_{1, rej} \cup S_{2, rej}$
		\item $\Theta_j$ and $\delta_j$ are defined as
	\end{itemize}
	\begin{enumerate}
		\item $\Theta$ is defined as a transition function: 
		\begin{equation}
		S \setminus (S_{1, acc} \cup S_{1, rej}) \times \Gamma \rightarrow U(\mathcal{H}(Q)) \cup P(\mathcal{H}(Q))
		\end{equation}
		\begin{enumerate}[(a)]
			\item if $\Theta_1(s, \sigma) \in U(\mathcal{H}(Q_1))$ corresponds to unitary operator over $\mathcal{H}(Q_1)$, then extend the $\Theta_1(Q)$ relating to $\Theta(s, \sigma) \ket{q_{2, i}}=\ket{q_{2, i}}$, for $0 \leq i \leq n_2$ and $\delta(s, \sigma)=\delta_1 (s, \sigma)$,
			\item if $\Theta_1(s, \sigma) \in P(\mathcal{H}(Q_1))$, which denotes projective measurement over $\mathcal{H}(Q_1)$ and measurement is represented by projector set $\{P_i\}$. Then, $\Theta(s, \sigma)$ denotes an orthogonal measurement on $\mathcal{H}(Q)$ represented by $\{P_i^{'}\} \cup \{I_1\}$ projection operators on $P(Q)=P(Q_1 \cup Q_2)$, where $P_i^{'}$ represent projection operators obtain by extending $P_i$ with $P_i^{'} \ket{q_{2, i}}$, for $0 \leq i \leq n_2$  and $I_1$ represents mapping of  projection operator to $\mathcal{H}(Q_2)$, i.e. an identity operator. 
		\end{enumerate}
		\item For any $s \in S_2$ and $\sigma \in \Sigma \cup \{\#, \$\}$
		\begin{enumerate}[(a)]
			\item if $\Theta_2(s, \sigma) \in U(\mathcal{H}(Q_2))$ corresponds to unitary operator over $\mathcal{H}(Q_2)$, then extend the $\Theta_2(Q)$ relating to $\Theta(s, \sigma) \ket{q_{1, i}}=\ket{q_{1, i}}$, for $0 \leq i \leq n_1$ and $\delta(s, \sigma)=\delta_2(s, \sigma)$,
			\item if $\Theta_2(s, \sigma) \in P(\mathcal{H}(Q_2))$, which denotes projective measurement over $\mathcal{H}(Q_2)$ and measurement is represented by projector set $\{P_k\}$, then $\Theta(s, \sigma)$ is an orthogonal measurement on $\mathcal{H}(Q)$ represented by $\{P_k^{'}\} \cup \{I_2\}$ projection operators on $P(Q)=P(Q_1 \cup Q_2)$, where $P_k^{'}$ extend $P_k$ to $\mathcal{H}(Q)$ by defining $P_k^{'} \ket{q_{1, k}}=0$, for $0 \leq k \leq n_1$. 
		\end{enumerate}
		\item For any $s \in S_{1, acc}$ and $\sigma \in \Sigma \cup \{\#, \$\}$
		\begin{enumerate}[(a)]
			\item  if $\sigma \neq \#$, then $\Theta(s, \sigma)=I$ and $\delta(s, \sigma)= (s, -1)$, where \textit{I} is an identity operator over $\mathcal{H}(Q)$,
			\item  else if $\sigma=\#$, then $\Theta(s, \sigma)$ represents an orthogonal measurement by projectors $\{\ket{q_{1, i}}\bra{q_{1, i}}\mid q_{1, i} \in Q_1\}$, $\delta(s, \sigma) (1, i)=(z_{1, i}, 0)$; $\delta(z_{1, i}, \#)=(s_{2, 0}, 0)$; $\Theta(z_{q, i}, \#)=U(q_{1, i}, q_{2, 0})$, where \textit{U} denotes a unitary operator over $\mathcal{H}(Q)$ satisfy $U\ket{q_{1, i}}=\ket{q_{2, 0}}$.
		\end{enumerate}
	\end{enumerate}
	Recall the languages $L_1=\{a^ng^*u^nc^* \mid n \geq 1\}$ and $L_2=\{a^*g^mu^*c^m \mid m \geq 1\}$. In respect of 2QCFA $M_p$ designed above, for any input string \textit{w} $\in \{a, u, g, c\}^*$, we have considered the following cases:
	\begin{itemize}
		\item if $w \in L_1 \cap L_2$, then at the end of the computation process $M_p$ enters a state $S_{acc}$ with probability atleast $(1-\epsilon_1)(1-\epsilon_2)$. Then, $L_p$ is said to be recognized by $M_p$ with probability atleast  $(1-\epsilon_1)(1-\epsilon_2)= (1-(\epsilon_1+\epsilon_2-\epsilon_1\epsilon_2))$.
		\item if $w \in L_1$, but $w \notin L_2$, then at the end of the computation $M_P$ enters a state $S_{1, acc} \times S_{2, rej}$ with probability atleast $(1-\epsilon_1)(1-\epsilon_2)$. Then, $L_p$ is said to be rejected by $M_p$ with probability atleast  $(1-\epsilon_1)(1-\epsilon_2)= (1-(\epsilon_1+\epsilon_2-\epsilon_1\epsilon_2))$.
		\item if $w \notin L_1$, then the state of $M_p$ is changed to $S_{1, rej}$ and it is said to be rejected with probability atleast $1-\epsilon_1$. 
	\end{itemize}
	Hence, if the languages $L_1$ and $L_2$ are said to be recognized by 2QCFA's $M_1$ and $M_2$ in polynomial time with error probabilities $\epsilon_1, \epsilon_2 \geq 0$, respectively. Then, $L_p=L_1 \cap L_2$ is said to be recognized by $M_p$ in polynomial time with error probability $\epsilon=\epsilon_1+\epsilon_2-\epsilon_1\epsilon_2$.
\end{proof}
\subsection{Dumbbell Structure}
Dumbbell shaped RNA structure is formed by analogy of DNA dumbbells comprised of two-helical stems closed by two hairpin loop structures. It plays an important role in analysis of local structures in DNA. The loops on its both sides  restrict its enzymatic cleavage and stabilize the duplex. It is successfully applied to transcriptional regulation \cite{39}. Fig 3 shows the representation of dumbbell shaped  RNA secondary structure.  
\begin{figure}[h]
	\centering
	\includegraphics[scale=0.65]{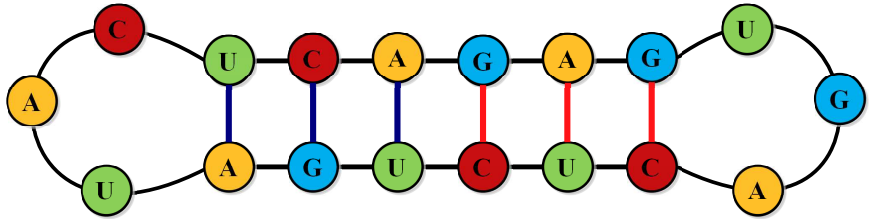}
	\caption{Representation of dumbbell bio-molecular structure}
\end{figure}
\begin{thm}
	A language $L_{d}=\{a^nu^ng^mc^m \mid n, m \geq 1\}$ representing dumbbell biomolecule structure can be recognized by 2QCFA with one-sided error probability in polynomial time. 
\end{thm}

\begin{proof}
	Ambainis and Watrous proved that a language $L_{eq}=\{a^nb^n \mid n \in \mathbb{N}\}$ can be recognized by 2QCFA in polynomial time, which can be recognized by 2PFA in exponential time \cite {1}. Similarly, it can be proved that a language $L_{d}=\{a^nu^ng^mc^m \mid n, m \geq 1\}$ can be recognized by 2QCFA $M_d$ in polynomial time with one-sided error probability. It consists of three phases. Firstly, it examines whether the input string is in form $a^{+}u^{+}g^{+}c^{+}$, if not, then the input string is said to be rejected. Otherwise, in second phase, 2QCFA simulates the initial part of string to determine whether $a^nu^n$ is in $L_d$, by using \textit{g} in the right side of \textit{u} as the right-end marker \$. If not, the computation is said to be rejected. Otherwise, in third phase, the 2QCFA finally checks $g^mc^m$ is in $L_d$ and \textit{u} is used as a left-end marker \#. If the number of \textit{g}'s and \textit{c}'s are equal, then it is said to be recognized with one-sided error probability in polynomial time, otherwise rejected.
\end{proof}
\section{Conclusion}
The enhancement in many existing computational approaches provides momentum to biological systems and quantum simulations at the gene expression levels. It helps to test new abstract approaches for considering RNA, DNA and protein sequences. Previous attempts to model the aforementioned RNA secondary structures used formal grammar and finite automata theory. In this paper, we focused on well-known structures of RNA such as hairpin loop, pseudoknot and dumbbell biomolecular structures and modeled them using two-way quantum finite automata with classical states. The crucial advantage of the quantum approach is that these secondary structures transcribed in formal languages takes exponential time for $L_h$ and polynomial time for $L_p$ and $L_d$, respectively. It has been shown that two-way quantum finite automata with classical states are more superior than its classical variants by using quantum part of finite size. For the future purpose, we will try to represent complex RNA structures in formal languages and model them using other quantum computational models.
\section*{Additional Information}

\textbf{Conflict of interest} The authors declare that they have no conflict of interest.
\section*{Acknowledgement}
S.Z. acknowledges support in part from the National Natural Science
Foundation of China (Nos. 61602532),  the Natural Science Foundation of Guangdong Province of China (No. 2017A030313378), and the Science and Technology Program of Guangzhou City of China (No. 201707010194).

\end{document}